\documentclass[a4paper,fleqn]{article}

\usepackage{amssymb,amsmath,amsthm}
\usepackage{graphicx}
\usepackage{hyperref}
\usepackage{pgfplots}
\pgfplotsset{compat=newest}

\newcommand{\cln}{\,:\,}
\newcommand{\ve}[1]{\mathbf v(#1)}

\theoremstyle{theorem}
\newtheorem{prop}{Proposition}
\theoremstyle{remark}
\newtheorem{remark}{Remark}

\begin{document}

\title{An implicitization-based solution to the minimal 4s/6r ToA problem using Cayley--Menger determinants}

\author{Evgeniy	Martyushev\\
South Ural State University\\
76 Lenin Avenue, Chelyabinsk 454080, Russia\\
{\tt martiushevev@susu.ru}
}

\date{}

\maketitle

\begin{abstract}
The paper introduces an efficient algebraic solver for the 4-sender/6-receiver (4s/6r) Time-of-Arrival (ToA) self-localization problem, which involves determining the relative positions of all receivers and senders given their pairwise distance measurements. The problem is addressed through a new parametrization combining Cayley--Menger determinants with an implicitization technique. The proposed algorithm proceeds in three steps. First, a $148\times 211$ Macaulay matrix is constructed from the coefficients of the original polynomial system. Second, PLU decomposition of this matrix yields a $63\times 63$ matrix pair. Finally, up to 38 real solutions are obtained via generalized eigendecomposition followed by a validation step. Experiments on synthetic noise-free data demonstrate that the proposed solver outperforms existing methods by approximately three orders of magnitude in numerical accuracy while achieving an average runtime of $1.3\times$ faster than the fastest alternative. Experiments on a real-world acoustic dataset confirm that, when integrated within a RANSAC framework, the solver provides a reliable initial guess for bundle adjustment refinement.

\vspace{0.3cm}

\end{abstract}

\section{Introduction}
\label{sec:intro}

The $m$-senders/$n$-receivers ($m$s/$n$r) Time-of-Arrival (ToA) self-localization problem involves determining the relative positions of $m$ senders $\mathbf s_i$ and $n$ receivers $\mathbf r_k$ in 3-space, given pairwise distance measurements $d_{ik}$ derived from signal propagation times, see Figure~\ref{fig:toa46}. This problem arises naturally in applications where absolute positions are unavailable, requiring the system to infer its own geometry solely from inter-element distances.

ToA based localization systems have found widespread adoption in various engineering domains including RF/ultra\-sound indoor localization~\cite{holm2009hybrid,mainetti2014survey,jia2015indoor,sakpere2017state,zafari2019survey}, anchor-free autonomous calibration of acoustic sensor networks~\cite{contini2012self,crocco2012closed,gaubitch2013auto,zhayida2015toa}, underwater navigation~\cite{eustice2007experimental,casey2007underwater,eustice2011synchronous,han2012localization,yi2015toa}, and swarm robotics coordination~\cite{han2015crlb,staudinger2021role,chen2022survey}. In all these applications, solving the self-localization problem efficiently and accurately is a fundamental task.

The problem becomes minimal in the sense of admitting a finite but non-zero number of solutions for generic values $d_{ik}$ when $m = 4$ and $n = 6$ (or vice versa, by symmetry)~\cite{stewenius2005grobner}. Another minimal case $m = n = 5$, not addressed in this paper, requires a single constraint on the distances $d_{ik}$~\cite{kuang2013complete}. The minimal 4s/6r configuration is of practical importance: it represents the smallest problem instance that yields isolated solutions, making it suitable as a hypothesis generator within robust estimation frameworks such as RANSAC~\cite{fischler1981random} for overdetermined ToA problems with many senders and receivers.

The minimal 4s/6r problem has been studied in prior works~\cite{stewenius2005grobner,kuang2013complete,larsson2017polynomial,larsson2020upgrade,martyushev2026automatic}. The state-of-the-art parametrization, introduced in~\cite{kuang2013complete}, formulates the minimal 4s/6r problem as a system of four cubic and one quartic polynomial equations in five variables. Its solution space decomposes into two components: a $1$-dimensional family of spurious roots without a feasible interpretation and a $0$-dimensional component consisting of $38$ points in complex space. Although this parametrization successfully reduced the problem to a compact polynomial system and enabled solver generation using automatic generators from~\cite{larsson2017efficient,larsson2017polynomial,martyushev2026automatic}, existing implementations have faced challenges in simultaneously achieving high numerical accuracy and computational efficiency.

\begin{figure}[ht]
\label{fig:toa46}
\centering 
\includegraphics[scale=0.35]{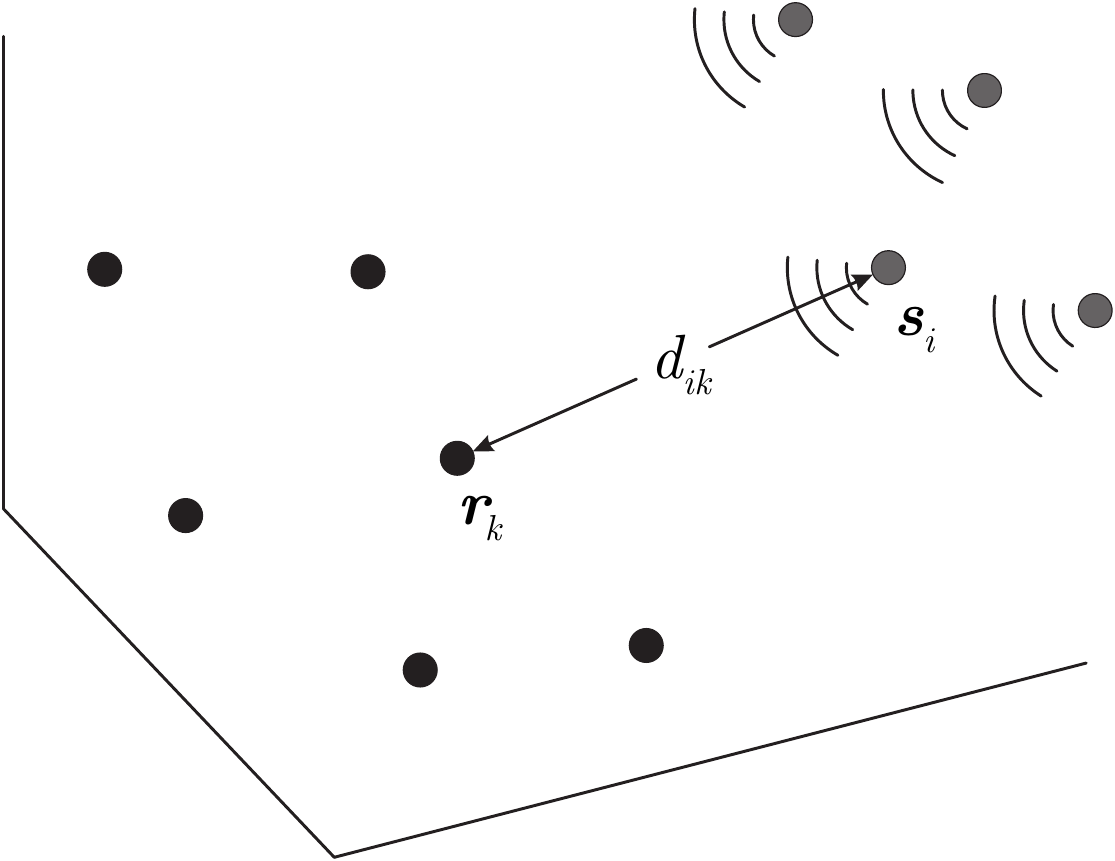}
\caption{4s/6r ToA self-localization problem}
\end{figure}

An alternative parametrization for the 4s/6r ToA problem is based on the implicitization technique. The idea of using implicitization for this problem was first suggested by Stew{\'{e}}nius~\cite{stewenius2005grobner}, who introduced the notion of an $11$-dimensional ``quadri-sonal vector'' and stated the existence of non-linear constraints on this vector. However, that work did not provide an explicit form of these constraints nor a practical solver based on this parametrization.

In this paper, we develop the implicitization approach for the 4s/6r ToA problem into an efficient practical solver. First, we construct an $11$-dimensional vector from Cayley--Menger determinants and apply implicitization to derive the complete set of polynomial constraints on its components. This yields a system of eight cubic and two quartic polynomial equations in five variables, whose solution space is $0$-dimensional and consists of exactly $38$ complex roots. For a given set of distance measurements, the system is solved efficiently using the method of elimination templates, an approach widely adopted in geometric computer vision and robotics~\cite{kukelova2008automatic,larsson2017efficient,larsson2018beyond,bhayani2020sparse,martyushev2022optimizing,martyushev2026automatic,martyushev2025forward}. The resulting solver is faster compared to state-of-the-art alternatives while also demonstrating significantly improved numerical accuracy.

The rest of the paper is organized as follows. Section~\ref{sec:prob} introduces the new parametrization for the 4s/6r ToA problem, reducing it to a system of ten polynomial equations in five variables. Section~\ref{sec:alg} details the proposed elimination template based algorithm for solving the ToA problem. In Section~\ref{sec:exper}, we evaluate the solver's performance, first on synthetic noise-free data to compare its robustness and speed against existing approaches, second on synthetic noisy data in the presence of outliers, utilizing the proposed solver as a hypothesis generator for RANSAC, and third on a real acoustic dataset. Section~\ref{sec:disc} discusses the results of the paper.

\section{Problem statement and new parametrization}
\label{sec:prob}

Consider the minimal 4s/6r ToA self-localization problem illustrated in Figure~\ref{fig:toa46}. Let $\mathbf s_i \in \mathbb R^3$ for $i = 1, \ldots, 4$ denote the positions of the senders, and $\mathbf r_k \in \mathbb R^3$ for $k = 1, \ldots, 6$ the positions of the receivers. The Euclidean distance between sender $\mathbf s_i$ and receiver $\mathbf r_k$ is $\|\mathbf s_i - \mathbf r_k\|$. The measured distance $d_{ik} = ct_{ik}$ is obtained from the time-of-arrival $t_{ik}$ and the known signal propagation speed $c$ (e.g., speed of light for RF signals or speed of sound for acoustic signals). This yields the system of quadratic equations:
\begin{equation}
\label{eq:ToA}
\|\mathbf s_i - \mathbf r_k\|^2 = d_{ik}^2, \quad i = 1, \ldots, 4, \quad k = 1, \ldots, 6.
\end{equation}
Due to the $6$ degree-of-freedom Euclidean ambiguity, the number of unknowns in system~\eqref{eq:ToA} reduces to $24$. Thus,~\eqref{eq:ToA} is a square system of $24$ equations in $24$ unknowns. Despite the simplicity of formulation~\eqref{eq:ToA}, developing computationally efficient and numerically accurate solvers for arbitrary distance measurements $d_{ik}$ remains a challenging problem.

\begin{remark}
\label{rem:genpos}
Throughout this paper, we assume that all senders and receivers are in general position. In particular, no two senders or two receivers coincide and both sender and receiver configurations span the full 3-space.
\end{remark}

\subsection{Number of solutions and symmetry}

It was conjectured in~\cite{stewenius2005grobner} that the 4s/6r ToA problem has $38$ complex solutions up to Euclidean symmetries. This can be verified computationally as follows. Fix the Euclidean ambiguity by setting
\begin{equation}
\mathbf s_1 = \mathbf 0, \quad (\mathbf s_2)_2 = (\mathbf s_2)_3 = (\mathbf s_3)_3 = 0,
\end{equation}
where $(\mathbf s_i)_j$ denotes the $j$th coordinate of $\mathbf s_i$. For randomly chosen rational distances $d_{ik}$ (not necessarily corresponding to a realizable configuration), consider the polynomial ideal
\begin{equation}
I = \langle\|\mathbf s_i - \mathbf r_k\|^2 - d_{ik}^2 \cln \forall i,k\rangle.
\end{equation}
Computing a Gr\"obner basis for $I$ over a sufficiently large finite field yields $304$ standard monomials (those not divisible by leading terms of $I$). This implies~\cite{Cox-IVA-2015} that system~\eqref{eq:ToA} has $304$ complex roots counting multiplicities. However, the set of roots exhibits the following $8$-fold symmetry. If $(\mathbf s_1, \ldots, \mathbf s_4, \mathbf r_1, \ldots, \mathbf r_6)$ is a root, then so is
\[
\begin{bmatrix}\pm 1 & 0 & 0 \\ 0 & \pm 1 & 0 \\ 0 & 0 & \pm 1 \end{bmatrix} \begin{bmatrix} \mathbf s_1 & \ldots & \mathbf s_4 & \mathbf r_1 & \ldots & \mathbf r_6 \end{bmatrix}
\]
for any combination of $+$'s and $-$'s. Consequently, the number of geometrically distinct solutions is $304/8 = 38$.

\subsection{New parametrization}

Consider the $4$-dimensional simplex embedded in $3$-space with vertices $\mathbf s_1, \ldots, \mathbf s_4$ and $\mathbf r_k$ for a certain $k$. Let $D_{ik} = d_{ik}^2$ denote the squared distances between senders and receivers, and $L_{ij} = \|\mathbf s_i - \mathbf s_j\|^2$ the squared inter-sender distances. Given $10$ squared edge lengths $L_{12}, \ldots, L_{34}, D_{1k}, \ldots, D_{4k}$ of the $4$-simplex, the Cayley--Menger determinant is defined as~\cite{berger2009geometry}:
\begin{equation}
\Gamma_k = \begin{vmatrix}
0 & 1 & 1 & 1 & 1 & 1\\
1 & 0 & L_{12} & L_{13} & L_{14} & D_{1k}\\
1 & L_{12} & 0 & L_{23} & L_{24} & D_{2k}\\
1 & L_{13} & L_{23} & 0 & L_{34} & D_{3k}\\
1 & L_{14} & L_{24} & L_{34} & 0 & D_{4k}\\
1 & D_{1k} & D_{2k} & D_{3k} & D_{4k} & 0
\end{vmatrix}.
\end{equation}
Since the hypervolume of the 4-simplex is up to a scale the squared Cayley--Menger determinant, it follows that $\Gamma_k = 0$. Through multivariate Taylor expansion, the equation $\Gamma_k = 0$ can be rewritten in the form
\begin{equation}
\label{eq:lineq}
\mathbf B_k^\top \mathbf T = 0,
\end{equation}
where the superscript $\top$ indicates matrix transpose and the coefficient vector $\mathbf B_k$ is given by
\begin{equation}
\mathbf B_k =
\begin{bmatrix}
\Delta_{2k}^2 &
\Delta_{3k}^2 &
\Delta_{4k}^2 &
2\Delta_{2k}\Delta_{3k} &
2\Delta_{2k}\Delta_{4k} &
2\Delta_{3k}\Delta_{4k} &
2\Delta_{2k} &
2\Delta_{3k} &
2\Delta_{4k} &
4D_{1k} &
1
\end{bmatrix}^\top
\end{equation}
with $\Delta_{ik} = D_{ik} - D_{1k}$. The vector $\mathbf T \in \mathbb R^{11}$ has the components:
\begin{equation}
\label{eq:vecT}
\small
\mathbf T = \begin{bmatrix}
-L_{13}^2 - L_{14}^2 - L_{34}^2 + 2L_{13}L_{14} + 2L_{13}L_{34} + 2L_{14}L_{34}\\
-L_{12}^2 - L_{14}^2 - L_{24}^2 + 2L_{12}L_{14} + 2L_{12}L_{24} + 2L_{14}L_{24}\\
-L_{12}^2 - L_{13}^2 - L_{23}^2 + 2L_{12}L_{13} + 2L_{12}L_{23} + 2L_{13}L_{23}\\
L_{14}^2 + L_{12}L_{13} + L_{24}L_{34} - L_{12}L_{14} - L_{12}L_{34} - L_{13}L_{14} - L_{13}L_{24} - L_{14}L_{24} - L_{14}L_{34} + 2L_{14}L_{23}\\
L_{13}^2 + L_{12}L_{14} + L_{23}L_{34} - L_{12}L_{13} - L_{12}L_{34} - L_{13}L_{14} - L_{14}L_{23} - L_{13}L_{23} - L_{13}L_{34} + 2L_{13}L_{24}\\
L_{12}^2 + L_{13}L_{14} + L_{23}L_{24} - L_{12}L_{13} - L_{13}L_{24} - L_{12}L_{14} - L_{14}L_{23} - L_{12}L_{23} - L_{12}L_{24} + 2L_{12}L_{34}\\
L_{13}^2L_{24} + L_{14}^2L_{23} + L_{34}^2L_{12} - L_{12}L_{13}L_{34} - L_{12}L_{14}L_{34} - L_{13}L_{14}L_{23} - L_{13}L_{14}L_{24} - L_{13}L_{24}L_{34} - L_{14}L_{23}L_{34} + 2L_{13}L_{14}L_{34}\\
L_{14}^2L_{23} + L_{12}^2L_{34} + L_{24}^2L_{13} - L_{12}L_{13}L_{24} - L_{13}L_{14}L_{24} - L_{12}L_{14}L_{23} - L_{12}L_{14}L_{34} - L_{12}L_{24}L_{34} - L_{14}L_{23}L_{24} + 2L_{12}L_{14}L_{24}\\
L_{12}^2L_{34} + L_{13}^2L_{24} + L_{23}^2L_{14} - L_{12}L_{14}L_{23} - L_{13}L_{14}L_{23} - L_{12}L_{13}L_{24} - L_{12}L_{13}L_{34} - L_{12}L_{23}L_{34} - L_{13}L_{23}L_{24} + 2L_{12}L_{13}L_{23}\\
T_{10}\\
-L_{12}^2L_{34}^2 - L_{13}^2L_{24}^2 - L_{14}^2L_{23}^2 + 2L_{12}L_{13}L_{24}L_{34} + 2L_{12}L_{14}L_{23}L_{34} + 2L_{13}L_{14}L_{23}L_{24}
\end{bmatrix},
\end{equation}
where
\begin{equation}
\label{eq:T10}
T_{10} = -\frac{1}{2}\begin{vmatrix}
0 & 1 & 1 & 1 & 1\\
1 & 0 & L_{12} & L_{13} & L_{14}\\
1 & L_{12} & 0 & L_{23} & L_{24}\\
1 & L_{13} & L_{23} & 0 & L_{34}\\
1 & L_{14} & L_{24} & L_{34} & 0
\end{vmatrix}.
\end{equation}
Thus, $\mathbf T$ consists of polynomial functions of the inter-sender squared distances $L_{ij}$ and $T_{10}$ is up to a scale the squared volume of the $3$-simplex with the vertices $\mathbf s_1, \ldots, \mathbf s_4$. According to Remark~\ref{rem:genpos}, we must have $T_{10} \neq 0$.

\begin{remark}
Stew{\'{e}}nius in~\cite{stewenius2005grobner} introduced an analogous $11$-dimensional ``quadri-sonal vector'' and stated the existence of non-linear constraints that can be found using implicitization. However, neither an explicit form of these constraints nor a corresponding solver was provided in~\cite{stewenius2005grobner}.
\end{remark}

The implicitization of $\mathbf T$ involves computing generators of the elimination ideal
\begin{equation}
J = \langle T_i - t_i \cln i = 1, \ldots, 11 \rangle \cap \mathbb Q[t_1, \ldots, t_{11}],
\end{equation}
where $T_i$ are components of $\mathbf T$ and $t_i$ are formal variables.

\begin{prop}
\label{prop:elim}
The ideal $J$ is generated by $8$ cubic, $12$ quartic, and $12$ quintic polynomials in $t_1, \ldots, t_{11}$.
\end{prop}

\begin{proof}
Let $S = \{T_i - t_i \cln i = 1, \ldots, 11\}$. The elimination can be performed with the following \texttt{Macaulay2}~\cite{macaulay2} code:
\begin{verbatim}
R = QQ[L12,L13,L14,L23,L24,L34,t_1..t_11];
I = ideal(S);
J = eliminate({L12,L13,L14,L23,L24,L34}, I);
mingens J
\end{verbatim}
The result is the claimed set of generators. The eight cubic and two of the quartic generators are listed explicitly below in Eqs.~\eqref{eq:eq1}--\eqref{eq:eq10}.
\end{proof}

Let $\mathbf t = \begin{bmatrix}t_1 & \ldots & t_{11}\end{bmatrix}^\top$. Each receiver $\mathbf r_k$ provides a linear constraint on $\mathbf t$ of type~\eqref{eq:lineq}. Thus, with $6$ receivers, $\mathbf t$ can be parametrized in the form
\begin{equation}
\label{eq:vect}
\mathbf t = \frac{v}{z}\, \mathbf t^{(1)} + \frac{w}{z}\, \mathbf t^{(2)} + \frac{x}{z}\, \mathbf t^{(3)} + \frac{y}{z}\, \mathbf t^{(4)} + \frac{1}{z}\, \mathbf t^{(5)},
\end{equation}
where $v, w, x, y, z$ are unknowns and $\{\mathbf t^{(1)}, \ldots, \mathbf t^{(5)}\}$ forms a basis for the null-space of the stacked coefficient matrix $\begin{bmatrix}\mathbf B_1 & \ldots & \mathbf B_6\end{bmatrix}^\top$. The parametrization~\eqref{eq:vect} is preferred over the more direct form $\mathbf t = v\, \mathbf t^{(1)} + w\, \mathbf t^{(2)} + x\, \mathbf t^{(3)} + y\, \mathbf t^{(4)} + z\, \mathbf t^{(5)}$ for two reasons: (i) the latter introduces a spurious multiple zero solution, and (ii) it leads to a less computationally efficient solver.

Proposition~\ref{prop:elim} specifies $32$ polynomial constraints on the vector $\mathbf t$. While using all of these constraints would minimize both the elimination template size and the number of false roots, the computational cost of deriving the coefficients after substituting~\eqref{eq:vect} becomes prohibitive. As a practical compromise, we use only the following $10$ constraints for our solver:
\begin{itemize}
\item all $8$ cubic constraints:
\begin{align}
\label{eq:eq1}
(t_4t_6 - t_2t_5)t_7 + (t_4t_5 - t_1t_6)t_8 + (t_1t_2 - t_4^2)(t_9 - t_{10}) &= 0,\\
(t_5t_6 - t_3t_4)t_7 + (t_1t_3 - t_5^2)(t_8 - t_{10}) + (t_4t_5 - t_1t_6)t_9 &= 0,\\
(t_2t_3 - t_6^2)(t_7 - t_{10}) + (t_5t_6 - t_3t_4)t_8 + (t_4t_6 - t_2t_5)t_9 &= 0,\\
(t_2t_3 - t_6^2)t_1 + (t_1t_3 - t_5^2)t_4 + (t_1t_2 - t_4^2)t_5 - 4t_7t_{10} &= 0,\\
(t_1t_3 - t_5^2)t_2 + (t_2t_3 - t_6^2)t_4 + (t_1t_2 - t_4^2)t_6 - 4t_8t_{10} &= 0,\\
(t_1t_2 - t_4^2)t_3 + (t_2t_3 - t_6^2)t_5 + (t_1t_3 - t_5^2)t_6 - 4t_9t_{10} &= 0,\\
(t_2t_3 - t_6^2)t_1 + (t_5t_6 - t_3t_4)t_4 + (t_4t_6 - t_2t_5)t_5 - 4t_{10}^2 &= 0,\\
\label{eq:eq8}
(t_2t_3 - t_6^2)t_7 + (t_1t_3 - t_5^2)t_8 + (t_1t_2 - t_4^2)t_9 - 4t_{10}t_{11} &= 0;
\end{align}
\item $2$ of $12$ quartic constraints:
\begin{align}
\label{eq:eq9}
t_{10}^2(t_1t_2 - t_4^2) - t_{10}^2(t_2t_3 - t_6^2) + t_9t_{10}(t_4t_6 - t_2t_5) - t_9t_{10}(t_1t_2 - t_4^2) + t_8t_{10}(t_5t_6 - t_3t_4) \notag\\+ t_8t_{10}(t_1t_3 - t_5^2) + t_2t_{11}(t_4t_5 - t_1t_6) - t_1t_{11}(t_4t_6 - t_2t_5) + t_2t_5^2t_{11} - t_1t_6^2t_{11} + t_2t_3t_7^2 \notag\\+ t_2t_6t_7^2 - t_1t_3t_8^2 - t_1t_5t_8^2 + t_1t_2t_8t_9 - t_1t_2t_7t_9 + t_1t_4t_8t_9 - t_2t_4t_7t_9 + t_1t_6t_7t_8 - t_2t_5t_7t_8 \notag\\+ 2t_1t_6t_8t_9 - 2t_2t_5t_7t_9 + 2t_2t_5t_7t_{10} - 2t_4t_5t_8t_{10} = 0,&\\
\label{eq:eq10}
t_{10}^2(t_1t_3 - t_5^2) - t_{10}^2(t_2t_3 - t_6^2) + t_8t_{10}(t_5t_6 - t_3t_4) - t_8t_{10}(t_1t_3 - t_5^2) + t_9t_{10}(t_4t_6 - t_2t_5) \notag\\+ t_9t_{10}(t_1t_2 - t_4^2) + t_3t_{11}(t_4t_5 - t_1t_6) - t_1t_{11}(t_5t_6 - t_3t_4) + t_3t_4^2t_{11} - t_1t_6^2t_{11} + t_2t_3t_7^2 \notag\\+ t_3t_6t_7^2 - t_1t_2t_9^2 - t_1t_4t_9^2 + t_1t_3t_8t_9 - t_1t_3t_7t_8 + t_1t_5t_8t_9 - t_3t_4t_7t_9 + t_1t_6t_7t_9 - t_3t_5t_7t_8 \notag\\+ 2t_1t_6t_8t_9 - 2t_3t_4t_7t_8 + 2t_3t_4t_7t_{10} - 2t_4t_5t_9t_{10} = 0.&
\end{align}
\end{itemize}

Substituting the parametrization~\eqref{eq:vect} into constraints~\eqref{eq:eq1}--\eqref{eq:eq10} and applying appropriate scaling factors ($z^3$ for equations~\eqref{eq:eq1}--\eqref{eq:eq8} and $z^4$ for equations~\eqref{eq:eq9}--\eqref{eq:eq10}) yields a system of $10$ polynomial equations
\begin{equation}
\label{eq:polsys}
f_i = 0, \quad i = 1, \ldots, 10
\end{equation}
in the $5$ variables $v, w, x, y, z$. The system is generically $0$-dimensional and possesses exactly $38$ complex roots.

\section{Description of the algorithm}
\label{sec:alg}

To solve system~\eqref{eq:polsys} efficiently for a given set of distance measurements, we employ the state-of-the-art automatic solver generator described in~\cite{martyushev2026automatic}. This generator takes as input an instance of a polynomial system with rational coefficients and returns the following data: (i) an action Laurent monomial (typically a variable or its reciprocal); (ii) an ordered $s$-tuple of monomial sets $A = (A_1, \ldots, A_s)$, where $s$ is the number of equations; (iii) a set of basic monomials $\mathcal B$. These data are used to expand the original polynomial system with monomial multiples. The Macaulay (coefficient) matrix of the expanded system is known as the elimination template. Standard linear algebra decompositions applied to this template recover the roots of the original polynomial system, see~\cite{kukelova2008automatic,byrod2009fast,larsson2017efficient,larsson2018beyond,martyushev2022optimizing,martyushev2026automatic,martyushev2025forward} for details.

Applying the generator~\cite{martyushev2026automatic} to system~\eqref{eq:polsys} yields the following data. The action monomial (variable) is $v$. The ordered $10$-tuple of monomial sets $A = (A_1, \ldots, A_{10})$ is defined by
\begin{equation}
\begin{split}
A_1 = A_2 = A_3 = A_4 = A_5 &= \{1, z, y, x, w, v, yz, xz, wz, vz, y^2, xy, wy, vy, x^2, wx, vx, w^2, vw, v^2\},\\
A_6 = A_7 &= \{1, z, x, w, v, xz, wz, vz, x^2, wx, vx, w^2, vw, v^2\},\\
A_8 &= \{1, z, x, w, v, xz, wz, vz\},\\
A_9 = A_{10} &= \{1, z, y, x, w, v\}.
\end{split}
\end{equation}
The set of basic monomials is
\setlength{\arraycolsep}{3pt}
\begin{equation}
\label{eq:vB}
\begin{array}{llcccccccccr}
\mathcal B = \{\!\!\! & vy^3, & wy^3, & y^3x, & y^4, & w^3z, & v^2zx, & vzwx, & w^2zx, & vzx^2, & wzx^2, & x^3z, \\
& v^2zy, & vzwy, & w^2zy, & vzxy, & wzxy, & x^2zy, & vzy^2, & wzy^2, & zy^2x, & zy^3, & x^2w, \\
& x^3, & yv^2, & ywv, & yw^2, & yxv, & yxw, & yx^2, & y^2v, & y^2w, & xy^2, & y^3, \\
& zv^2, & zvw, & zw^2, & zvx, & zwx, & zx^2, & zvy, & zwy, & zxy, & zy^2, & v^2, \\
& vw, & w^2, & vx, & wx, & x^2, & vy, & wy, & xy, & y^2, & zv, & zw, \\
& zx, & zy, & v, & w, & x, & y, & z, & 1\}. & & &
\end{array}
\end{equation}
Its cardinality is $\#\mathcal B = 63$, meaning that the solver outputs $63$ candidate roots, of which $25$ roots are spurious.

The expanded set of polynomials is defined as
\begin{equation}
F = \{m \cdot f_i \cln \forall m \in A_i,\, i = 1, \ldots, 10\}.
\end{equation}
Let $\mathcal U$ denote the support of $F$, i.e., the set of all monomials appearing in at least one polynomial from $F$. Its cardinality is $\#\mathcal U = 211$. We partition $\mathcal U$ into three pairwise disjoint subsets $\mathcal U = \mathcal E \cup \mathcal R \cup \mathcal B$, where
\begin{equation}
\label{eq:vRvE}
\mathcal R = \{v m \cln m \in \mathcal B\} \setminus \mathcal B,\qquad
\mathcal E = \mathcal U \setminus (\mathcal R \cup \mathcal B).
\end{equation}
The sets $\mathcal R$ and $\mathcal E$ are called the sets of reducible and excessive monomials, respectively~\cite{byrod2009fast}. Their cardinalities are $\#\mathcal R = 41$ and $\#\mathcal E = 107$.

For a finite set of monomials $\mathcal A \subseteq \mathcal U$, we denote by $\ve{\mathcal A}$ the vector of monomials in $\mathcal A$ arranged according to a fixed total ordering on $\mathcal U$.

The $148\times 211$ Macaulay matrix $M$ is defined so that its $(i, j)$th entry is the coefficient of the $i$th polynomial in $F$ at the $j$th monomial in $\ve{\mathcal U}$. Consequently,
$
M \ve{\mathcal U} = \mathbf 0
$
represents the polynomial system $F = 0$ in vector form.

The matrix $M$ inherits a block structure
$
M =
\begin{bmatrix}
M_{\mathcal E} & M_{\mathcal R} & M_{\mathcal B}
\end{bmatrix},
$
where $M_{\mathcal E}$ is a $148\times 107$ sparse submatrix. Applying LU decomposition with partial pivoting on $M_{\mathcal E}$ gives $M_{\mathcal E} = PLU$, where $P$ is a permutation matrix. Then
\begin{equation}
\label{eq:PLM}
(PL_1)^{-1} M =
\begin{bmatrix}
U & * & * \\
0 & C_{\mathcal R} & C_{\mathcal B}
\end{bmatrix},
\end{equation}
where $C_{\mathcal R}$ is a full-rank $41 \times 41$ matrix, $C_{\mathcal B}$ is a $41 \times 63$ matrix, $L_1 = \begin{bmatrix}L & \begin{matrix} 0\\ I\end{matrix} \end{bmatrix}$ is a $148 \times 148$ lower-triangular matrix, and $I$ is the $41 \times 41$ identity matrix. The matrix $PL_1$ is invertible since the main diagonal of $L_1$ consists of $1$'s.

From~\eqref{eq:PLM} we obtain $C_{\mathcal R} \ve{\mathcal R} = - C_{\mathcal B} \ve{\mathcal B}$, which leads to a generalized eigenproblem:
\begin{equation}
\label{eq:T0T1}
v T_1 \ve{\mathcal B} = T_0 \ve{\mathcal B},
\end{equation}
where the $63\times 63$ matrices $T_0$ and $T_1$ contain only $0$, $1$, and the entries of $-C_{\mathcal B}$ and $C_{\mathcal R}$, see~\cite{martyushev2025forward} for details of their construction.

Let $\mathbf u_p$ be the $p$th generalized eigenvector of the matrix pair $(T_0, T_1)$, with $(\mathbf u_p)_i$ representing its $i$th component. Then the $p$th candidate root of~\eqref{eq:polsys} is given by
\begin{equation}
\label{eq:pt}
\hat{v}_p = \frac{(\mathbf u_p)_{58}}{(\mathbf u_p)_{63}}, \quad \hat{w}_p = \frac{(\mathbf u_p)_{59}}{(\mathbf u_p)_{63}}, \quad \hat{x}_p = \frac{(\mathbf u_p)_{60}}{(\mathbf u_p)_{63}}, \quad \hat{y}_p = \frac{(\mathbf u_p)_{61}}{(\mathbf u_p)_{63}}, \quad \hat{z}_p = \frac{(\mathbf u_p)_{62}}{(\mathbf u_p)_{63}}.
\end{equation}
The corresponding vector $\hat{\mathbf t}_p$ is obtained from the $5$-tuple $(\hat{v}_p, \hat{w}_p, \hat{x}_p, \hat{y}_p, \hat{z}_p)$ by Eq.~\eqref{eq:vect}.

Thus we compute $63$ candidate roots, and $25$ false roots are eliminated from them through the following validation procedure. Let $\hat{t}_i$ be the $i$th component of $\hat{\mathbf t}_p$ and $T_i$ the $i$th component of the vector $\mathbf T$ defined in Eq.~\eqref{eq:vecT}. The squared distances $L_{ij}$ are recovered in closed form using the first six components of $\hat{\mathbf t}_p$ and $\hat{t}_{10}$:
\begin{align}
L_{12} &= \alpha\bigl(\hat{t}_2\hat{t}_3 - \hat{t}_6^2\bigr),\\
L_{13} &= \alpha\bigl(\hat{t}_1\hat{t}_3 - \hat{t}_5^2\bigr),\\
L_{14} &= \alpha\bigl(\hat{t}_1\hat{t}_2 - \hat{t}_4^2\bigr),\\
L_{23} &= L_{12} + L_{13} - 2\alpha\bigl(\hat{t}_5\hat{t}_6 - \hat{t}_3\hat{t}_4\bigr),\\
L_{24} &= L_{12} + L_{14} - 2\alpha\bigl(\hat{t}_4\hat{t}_6 - \hat{t}_2\hat{t}_5\bigr),\\
L_{34} &= L_{13} + L_{14} - 2\alpha\bigl(\hat{t}_4\hat{t}_5 - \hat{t}_1\hat{t}_6\bigr),
\end{align}
where $\alpha = -1/(4\hat{t}_{10})$. Since the four senders must span $3$-space, $\hat{t}_{10} \neq 0$ for true solutions. For each candidate root, we compute the relative error
$
\xi = \|\hat{\mathbf t}_p - \mathbf T\|/\|\hat{\mathbf t}_p\|
$
and identify the $38$ true roots as those with the smallest $\xi$ values.

Finally, the sender and receiver positions $\hat{\mathbf s}_i$ and $\hat{\mathbf r}_p$ are recovered by solving a trilateration problem using the computed $L_{ij}$ and the given distances $D_{ik}$. Complex solutions are discarded, and the Euclidean ambiguity is resolved by fixing
\begin{equation}
\hat{\mathbf s}_1 = \mathbf 0, \quad
(\hat{\mathbf s}_2)_2 = (\hat{\mathbf s}_2)_3 = (\hat{\mathbf s}_3)_3 = 0, \quad
(\hat{\mathbf s}_2)_1 > 0, \quad (\hat{\mathbf s}_3)_2 > 0, \quad (\hat{\mathbf s}_4)_3 > 0.
\end{equation}

\section{Experiments}
\label{sec:exper}

In this section, we validate the MATLAB implementation of the proposed algorithm on both synthetic and real-world data. All experiments were performed on a system with Intel Core i5-1155G7 processor.

We modelled $m$ senders and $n$ receivers, each uniformly distributed within a unit cube. The ground truth positions of the senders and receivers are represented by the $3$-vectors $\mathbf s_i$ and $\mathbf r_k$, respectively.

\subsection{Noise-free data}

First, we compared the proposed solver with state-of-the-art algorithms from~\cite{kuang2013complete,larsson2020upgrade,martyushev2026automatic} on noise-free data with $m = 4$ senders and $n = 6$ receivers. From the ground truth vectors $\mathbf s_i$ and $\mathbf r_k$, we computed the $24$ distances $d_{ik} = \|\mathbf s_i - \mathbf r_k\|$, which serve as the input to all solvers.

To evaluate numerical accuracy, we defined two error metrics:
\begin{align}
\label{eq:metric1}
\epsilon_1 &= \max\biggl(\frac{1}{mn}\sum_{i,k} \bigl(\|\hat{\mathbf s}_i - \hat{\mathbf r}_k\| - d_{ik}\bigr)^2\biggr)^{1/2},\\
\label{eq:metric2}
\epsilon_2 &= \min\biggl(\sum_{j > i} \bigl(\|\mathbf s_i - \mathbf s_j\| - \|\hat{\mathbf s}_i - \hat{\mathbf s}_j\|\bigr)^2 + \sum_{l > k} \bigl(\|\mathbf r_k - \mathbf r_l\| - \|\hat{\mathbf r}_k - \hat{\mathbf r}_l\|\bigr)^2\biggr)^{1/2},
\end{align}
where $\hat{\mathbf s}_i$ and $\hat{\mathbf r}_k$ are the estimated positions. The maximum in~\eqref{eq:metric1} and the minimum in~\eqref{eq:metric2} are taken over all real solutions. The metric $\epsilon_1$ is sensitive to the presence of false solutions (high values indicate their presence), while $\epsilon_2$ detects the absence of a ground truth solution among the computed ones.

\begin{figure}[ht]
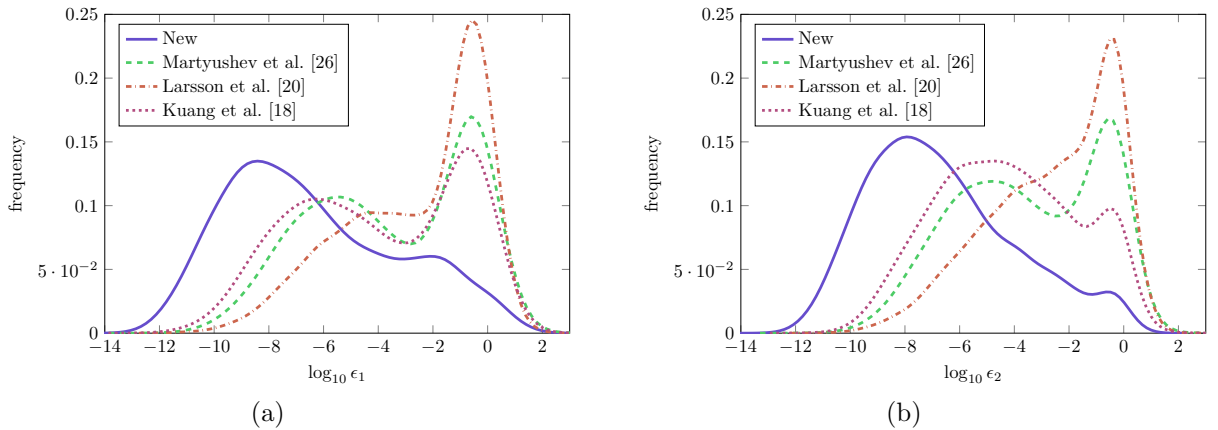

\centering
\begin{tabular}{cc}
\resizebox{0.47\textwidth}{!}{\input{numerror1}} & \resizebox{0.47\textwidth}{!}{\input{numerror2}}\\
(a) & (b)
\end{tabular}
\caption{Distributions over $10^4$ trials for error metrics from~\eqref{eq:metric1} and~\eqref{eq:metric2}}
\label{fig:numerr}
\end{figure}

\begin{table}[ht]
\centering
\begin{tabular}{llllllll}
\hline\\[-6pt]
Solver & Med. $\epsilon_1$ & Ave. $\epsilon_1$ & $\epsilon_1 > 0.1$ & Med. $\epsilon_2$ & Ave. $\epsilon_2$ & $\epsilon_2 > 0.1$ & Speed \\[3pt]
\hline\\[-6pt]
New & $7.4\times 10^{-8}$ & $4.0\times 10^{-7}$ & $9.85\%$ & $1.2\times 10^{-7}$ & $2.4\times 10^{-7}$ & $7.82\%$ & $4.5$ ms \\
Martyushev et al.~\cite{martyushev2026automatic} & $1.8\times 10^{-4}$ & $3.2\times 10^{-4}$ & $33.52\%$ & $7.1\times 10^{-4}$ & $3.8\times 10^{-4}$ & $30.64\%$ & $6.0$ ms \\
Larsson et al.~\cite{larsson2020upgrade} & $3.4\times 10^{-3}$ & $2.3\times 10^{-3}$ & $43.26\%$ & $1.7\times 10^{-2}$ & $2.5\times 10^{-3}$ & $40.12\%$ & $10.6$ ms \\
Kuang et al.~\cite{kuang2013complete} & $3.3\times 10^{-5}$ & $1.1\times 10^{-4}$ & $30.47\%$ & $8.8\times 10^{-5}$ & $5.8\times 10^{-5}$ & $20.74\%$ & $41.3$ ms \\[3pt]
\hline
\end{tabular}
\caption{Error distribution statistics and average runtime}
\label{tab:numerr}
\end{table}

Figure~\ref{fig:numerr} shows the error distribution histograms for both metrics, while Table~\ref{tab:numerr} provides detailed statistics, including median and mean values over $10^4$ trials, and average runtime. The proposed solver significantly outperforms existing methods in numerical robustness. Its median $\epsilon_1$ of $7.4\times 10^{-8}$ is about three orders of magnitude lower than that of the next best solver from~\cite{kuang2013complete}. In terms of speed, the new solver is about $1.3$ times faster than the fastest alternative. Most of its runtime ($53.2\%$) is spent on template construction, primarily due to computing the initial $10\times 85$ coefficient matrix.

While the total number of geometrically distinct complex roots for the 4s/6r ToA problem is $38$, it is currently unknown whether there exists an instance of the problem with exactly $38$ real solutions. In our experiments, the observed maximum number of reliably identified and well separated real solutions was~$20$.

\subsection{Noisy data with outliers}

In real-world applications, ToA measurements are often corrupted by noise and may contain outliers, such as invalid or missing measurements. To address this, a hypothesize-and-test framework like RANSAC~\cite{fischler1981random} can be used. The RANSAC procedure consists of the following main steps. First, multiple minimal subsets of $4$ senders and $6$ receivers are randomly sampled from the full set of $m \geq 4$ senders and $n \geq 6$ receivers. Next, the proposed minimal solver generates hypotheses for each subset. Each hypothesis is scored by classifying distance measurements as inliers or outliers. Finally, the winning hypothesis, selected based on the highest inlier count, is further refined through several iterations of the Levenberg--Marquardt algorithm.

We evaluated the performance of the proposed solver within the RANSAC framework for the two types of overdetermined scenarios:
\begin{itemize}
\item $m = 4$, $n \geq 7$;
\item $m \geq 5$, $n = 6$.
\end{itemize}

To simulate measurement noise, each true distance $d_{ik}$ was perturbed as $d_{ik}(1 + s_{ik})$, where $s_{ik}$ follows a normal distribution with zero mean and standard deviation $\sigma$. Outliers were introduced by replacing randomly selected distances $d_{ik}$ with uniformly distributed random values in the interval $[0, 3]$.

\begin{figure}[ht]
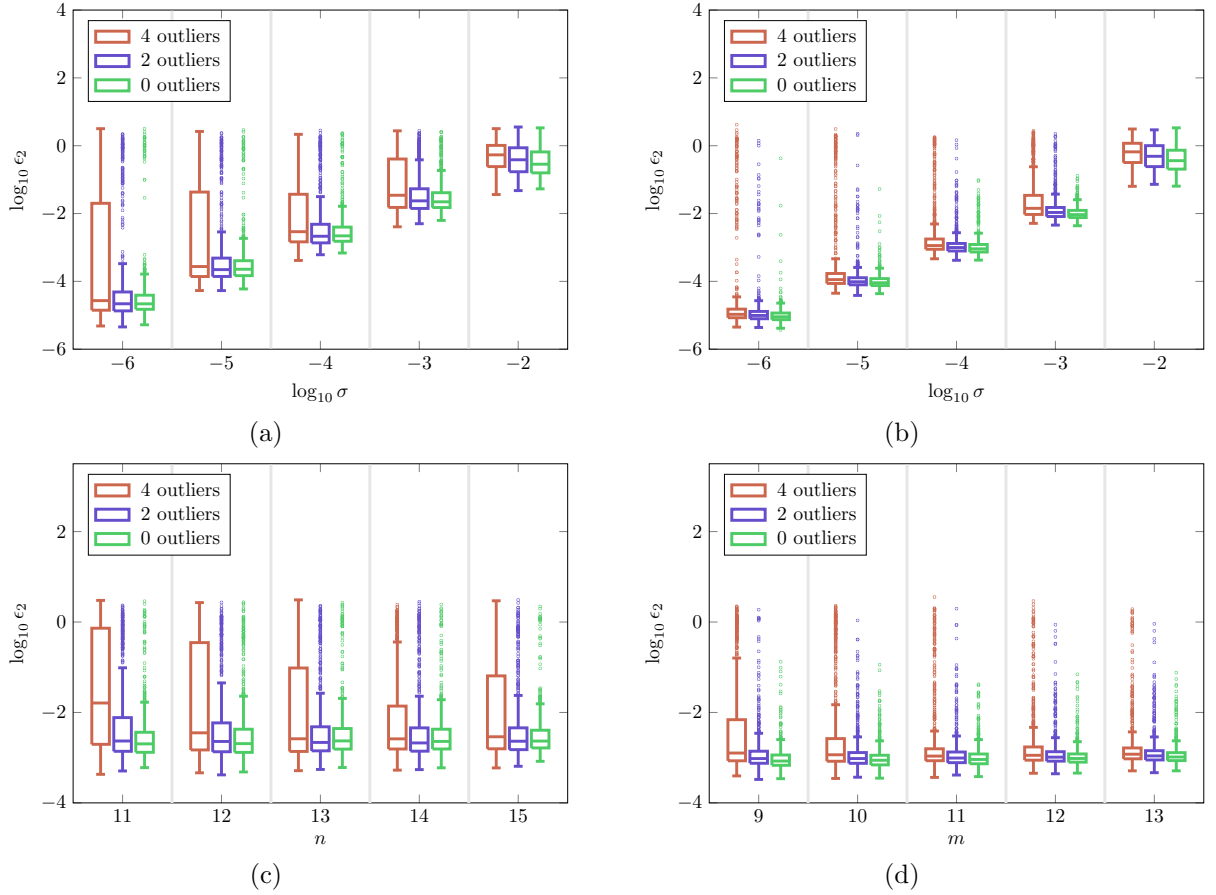

\centering
\begin{tabular}{cc}
\resizebox{0.47\textwidth}{!}{\input{boxplot1}} & \resizebox{0.47\textwidth}{!}{\input{boxplot3}}\\
(a) & (b)\\[5pt]
\resizebox{0.47\textwidth}{!}{\input{boxplot2}} & \resizebox{0.47\textwidth}{!}{\input{boxplot4}}\\
(c) & (d)
\end{tabular}
\caption{Performance on noisy data with outliers: error $\epsilon_2$ defined in~\eqref{eq:metric2} versus (a) noise level $\sigma$ ($m = 4$, $n = 13$); (b) noise level $\sigma$ ($m = 9$, $n = 6$); (c) number of receivers $n$ ($m = 4$, $\sigma = 10^{-4}$); (d) number of senders $m$ ($n = 6$, $\sigma = 10^{-4}$). Each boxplot summarizes the results from $10^3$ independent trials, displaying the interquartile range ($25\%$ to $75\%$) as a box, the median as an internal line, whiskers extending to $1.5\times$ the interquartile range, and dots indicating outliers}
\label{fig:ransac}
\end{figure}

Four experimental configurations were evaluated. In the first two tests, the number of senders and receivers was fixed at $m = 4$, $n = 13$ (Figure~\ref{fig:ransac}(a)) and at $m = 9$, $n = 6$ (Figure~\ref{fig:ransac}(b)), while the noise level $\sigma$ was varied from $10^{-6}$ to $10^{-2}$. The third test maintained a constant number of senders $m = 4$ and noise level $\sigma = 10^{-4}$, while varying the number of receivers $n$ from $11$ to $15$ (Figure~\ref{fig:ransac}(c)). Finally, the fourth test maintained a constant number of receivers $n = 6$ and noise level $\sigma = 10^{-4}$, while varying the number of senders $m$ from $9$ to $13$ (Figure~\ref{fig:ransac}(d)). In all tests, the number of outliers was set to $0$, $2$, and $4$. The results demonstrate that our RANSAC implementation with the new solver remains robust even with up to $4$ outliers and high noise levels.

\subsection{Real data}

We validated the proposed solver on the real-world \texttt{bassh2} dataset~\cite{zhayida2016automatic}, which contains eight channel audio recordings of a moving sound source in a reverberant room together with ground truth microphone positions and source trajectories. The dataset is provided together with the multi-module pipeline for automatic acoustic geometry calibration. We integrated our solver into this system, replacing only its geometric estimation module described in~\cite[Section 6]{zhayida2016automatic}. All other components, including signal processing, range‑difference extraction, offset estimation, and matching, remained unchanged.

For each time instant, the original pipeline first extracts range-difference measurements and performs offset estimation to obtain absolute distance estimates. These distances are then fed into our solver within a RANSAC framework to generate initial hypotheses for microphone and source positions. The best hypothesis is subsequently refined by the original system's bundle adjustment step.

\begin{figure}[ht]
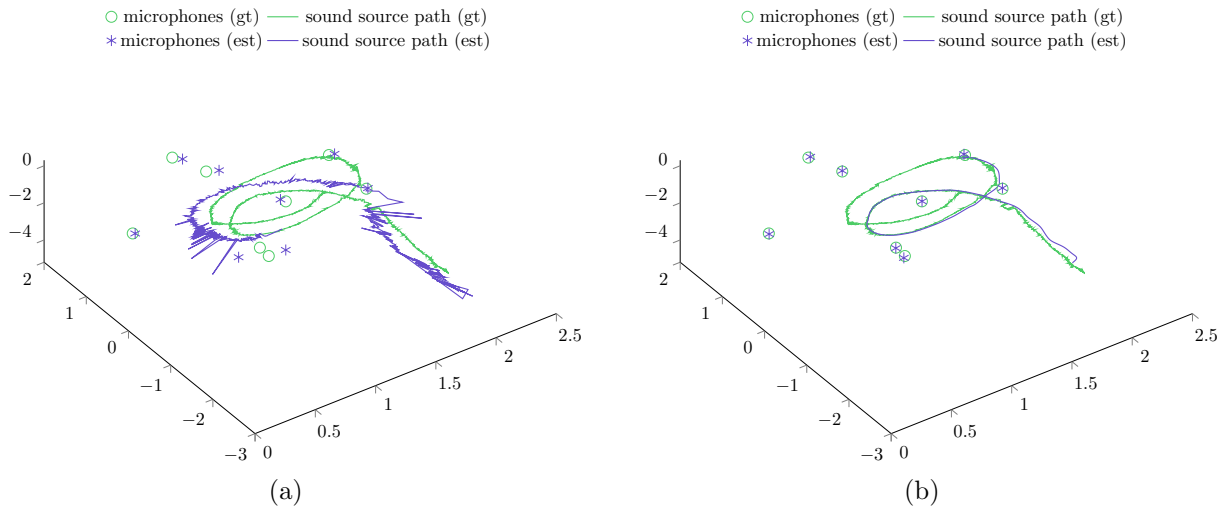

\centering
\begin{tabular}{cc}
\resizebox{0.47\textwidth}{!}{\input{real1}} & \resizebox{0.47\textwidth}{!}{\input{real2}}\\
(a) & (b)
\end{tabular}
\caption{Reconstructed sound source trajectory and microphone positions for the \texttt{bassh2} dataset~\cite{zhayida2016automatic}: (a) initial reconstruction (RMSE $= 0.1291$~m); (b) final result after bundle adjustment (RMSE $= 0.019$~m)}
\label{fig:real}
\end{figure}

With this hybrid approach, the initial reconstruction achieved a root mean square error (RMSE) of $0.1291$~m relative to ground truth (Figure~\ref{fig:real}(a)). After bundle adjustment, the RMSE dropped to $0.019$~m (Figure~\ref{fig:real}(b)).

\section{Conclusion}
\label{sec:disc}

The paper introduces a new parametrization for the minimal 4s/6r ToA problem, combining Cayley--Menger determinants with an implicitization technique. The resulting algebraic solver reduces the problem to ten polynomial equations in five variables and solves them via standard linear algebra decompositions (PLU$+$QZ).

Experiments on synthetic data demonstrate a substantial improvement in numerical accuracy, by orders of magnitude, over existing state-of-the-art methods. The solver achieves an average runtime of $4.5$~ms on standard hardware, enabling its effective use as a hypothesis generator within a RANSAC framework. Moreover, tests on a real dataset confirm that the RANSAC-based pipeline provides a reliable initial guess for bundle adjustment refinement.

The proposed approach may generalize to other minimal ToA problems (e.g., 5s/5r) and to broader sensor network calibration tasks. Future work should address degenerate configurations (e.g., coplanar senders/receivers) and extend the approach to Time-Difference-of-Arrival and hybrid calibration problems.

\subsection*{Data availability}

The MATLAB implementation of the proposed algorithm is available at \href{https://github.com/martyushev/t-d-oa}{https://github.com/martyushev/t-d-oa}.

\bibliographystyle{plain}
\bibliography{biblio}

\end{document}